\documentclass[letterpaper]{llncs}

\usepackage{amssymb}
\usepackage{graphicx}
\usepackage[ruled, vlined]{algorithm2e}
\usepackage{url}
\usepackage{subfig}
\usepackage{hyperref}

\newcommand{\cP}{{\cal P}}
\newcommand{\cR}{{\cal R}}
\newcommand{\cT}{{\cal T}}
\newcommand{\cA}{{\cal A}}
\newcommand{\wA}{{\widehat A}}

\newcommand{\wS}{{\widehat S}}
\newcommand{\wI}{{\widehat I}}
\newcommand{\cwT}{{\widehat {\cal T}}}

\begin{document}
\sloppy
\mainmatter

\title{A Competitive Analysis for Balanced Transactional Memory Workloads}

\titlerunning{A Competitive Analysis for Balanced Transactional Memory Workloads}

\author{Gokarna Sharma\and Costas Busch}

\authorrunning{G. Sharma and C. Busch}

\institute{Department of Computer Science\\ Louisiana State University,
Baton Rouge, LA 70803, USA\\ \email{\{gokarna,busch\}@csc.lsu.edu}}

\toctitle{A Competitive Analysis for Balanced Transactional Memory Workloads}
\tocauthor{G. Sharma and C. Busch}
\maketitle

\begin{abstract}
We consider transactional memory contention management
in the context of {\em balanced workloads},
where if a transaction is writing,
the number of write operations it performs is a constant fraction
of its total reads and writes.
We explore the theoretical performance boundaries of contention management
in balanced workloads from the worst-case perspective by
presenting and analyzing two new polynomial time contention management algorithms.
The first algorithm {\sf Clairvoyant}
is $O(\sqrt{s})$-competitive,
where $s$ is the number of shared resources.
This algorithm depends on explicitly knowing the conflict graph.
The second algorithm {\sf Non-Clairvoyant}
is $O(\sqrt{s} \cdot \log n)$-competitive,
with high probability,
which is only a $O(\log n)$ factor worse,
but does not require knowledge of the conflict graph,
where $n$ is the number of transactions.
Both of these algorithms are greedy.
We also prove that the performance of {\sf Clairvoyant} is tight,
since there is no polynomial time contention management algorithm that
is better than $O((\sqrt{s})^{1-\epsilon})$-competitive
for any constant $\epsilon > 0$,
unless {\sf NP$\subseteq$ZPP}.
To our knowledge, these results are significant improvements over the
best previously known $O(s)$ competitive ratio bound.
\end{abstract}

\section{Introduction}
\label{section:introduction}

The ability of multi-core architectures to increase application performance
depends on maximizing the utilization of the computing resources provided
by them and using multiple threads within applications.
These architectures present both an opportunity
and challenge for multi-threaded software.
The opportunity is that threads will be available to an unprecedented degree,
and the challenge is that more programmers will be exposed
to concurrency related synchronization problems
that until now were of concern only to a selected few.
Writing concurrent programs is a non-trivial task
because of the complexity of ensuring proper synchronization.
Conventional lock based synchronization (i.e., mutual exclusion) suffers from well known limitations,
so researchers considered non-blocking transactions as an alternative.
Herlihy and Moss \cite{Herlihy93transactionalmemory}
proposed Transactional Memory (TM),
as an alternative implementation of mutual exclusion,
which avoids many of the drawbacks of locks,
e.g., deadlock, reliance on the programmer to associate shared data with locks,
priority inversion, and failures of threads while holding locks.
Shavit and Touitou \cite{Shavit95} extended this idea to Software-only Transactional Memory (STM)
by proposing a novel software
method for supporting flexible transactional programming
of synchronization operations \cite{Her03,Har03,Har05}.

A transaction consists of a sequence of read and write operations
to a set of shared system resources (e.g. shared memory locations).
Transactions may conflict when they access the same shared resources.
If a transaction $T$ discovers that it conflicts with another transaction $T'$
(because they share a common resource),
it has two choices, it can give $T'$ a chance to commit by aborting itself,
or it can proceed and commit by forcing $T'$ to abort;
the aborted transaction then retries again until it eventually commits.
To solve the transaction scheduling problem efficiently,
each transaction consults with the {\em contention manager} module for which choice to make.
Dynamic STM (DSTM) \cite{Her03}, proposed for dynamic-sized data structures,
is the first STM implementation that uses a contention manager
as an independent module to resolve conflicts between two transactions and ensure progress.
Of particular interest are {\em greedy contention managers}
where a transaction restarts immediately after every abort.
As TM has been gaining attention,
several (greedy) contention managers
have been proposed in the literature \cite{Attiya10,Gue05a,Gue05b,carstm,Schneider09,Ramadan08}.
which have been assessed formally and experimentally
by specific benchmarks \cite{Sch05}.

A major challenge in guaranteeing progress through transactional contention
managers is to devise a policy which ensures that all transactions commit
in the shortest possible time.
The goal is to minimize the {\em makespan} which
is defined as the duration from the start of the schedule, i.e.,
the time when the first transaction is issued,
until all transactions commit.
The makespan of the transactional scheduling algorithm
can be compared to the makespan of an optimal off-line scheduling algorithm
to provide a {\em competitive ratio}.
The makespan and competitive ratio
primarily depend on the {\em workload} $-$ the set of transactions,
along with their arrival times, duration, and
resources they read and modify \cite{Bimodal09}.

The performance of some of the contention managers
has been analyzed formally in \cite{Bimodal09,Attiya10,Gue05a,Gue05b,Schneider09,Window10}
(the detailed description is given in Section \ref{section:related}).
The best known formal bound is provided in \cite{Attiya10}
where the authors give an $O(s)$ competitive ratio bound,
where $s$ is the number of shared resources.
When the number of resources $s$ increases, the performance degrades linearly.
A difficulty in obtaining better competitive ratios
is that the scheduling problem of $n$ concurrent transactions
is directly related to the vertex coloring problem
which is a hard problem to approximate \cite{Khot01}.
A natural question which we address here
is whether it is possible to obtain better competitive ratios.
As we show below,
it is indeed possible to obtain sub-linear competitive ratios
for balanced transaction workloads.

\subsection{Contributions}
In this paper, we study contention management
in the context of {\em balanced workloads}
which have better performance potential for transactional memory.
A balanced workload consists of a set of transactions
in which each transaction has the following property:
if the transaction performs write operations,
then the number of writes it performs is a constant fraction
of the total number of operations (read and writes) of the transaction.
The {\em balancing ratio} $\beta$ expresses the
ratio of write operations of a transaction
to the overall operations of the transaction.
The balancing ratio is bounded as $\frac{1} {s} \leq \beta \leq 1$,
since a writing transaction writes to at least one resource.
In balanced workloads $\beta = \Theta(1)$ for all the transactions which perform writes.
Balanced workloads can also include read-only transactions,
but we assume that there is at least one transaction that performs
writes, since otherwise the scheduling problem is trivial (no conflicts).

Balanced transaction workloads represent
interesting and practical transaction memory scheduling problems.
For example balanced workloads represent the case where
we have small sized transactions
each accessing a small (constant) number of resources,
where trivially $\beta = \Theta(1)$.
Other interesting scenarios are transaction workloads which are write intensive,
where transactions perform many writes, as for example in scientific computing
applications where transactions have to update large arrays.

We present two new polynomial time contention management algorithms
which are especially tailored for balanced workloads and analyze
their theoretical performance boundaries from the worst-case perspective.
The first algorithm, called {\sf Clairvoyant},
is $O \left (\ell \cdot \sqrt{\frac {s} {\beta} } \right)$-competitive
where $s$ is the number of shared resources,
and $\ell$ expresses the logarithm ratio of the longest to shortest execution times
of the transactions.
(The transaction execution time is the time it needs to commit uninterrupted
from the moment it starts.)
For balanced transaction workloads where $\beta = \Theta(1)$,
and when transaction execution times are
close to each other, i.e. $\ell = O(1)$,
Algorithm {\sf Clairvoyant} is $O(\sqrt{s})$-competitive.
This algorithm is greedy and has the pending commit property
(where at least one transaction executes uninterrupted each time).
However, it depends on assigning priorities to the transactions
based on the explicit knowledge
of the transaction conflict graph which evolves
while the execution of the transactions progresses.
It also assumes that each transaction knows how long is its execution time
and how many resources it accesses.

The second algorithm, called {\sf Non-Clairvoyant},
is $O \left (\ell \cdot \sqrt{\frac {s} {\beta}} \cdot \log n \right)$-competitive,
with high probability (at least $1- \frac{1} {n}$),
where $n$ is the number of transactions concurrently
executing in $n$ threads.
For balanced transaction workloads, where $\beta = \Theta(1)$,
and when transaction execution times are close to each other,
i.e. $\ell = O(1)$,
Algorithm {\sf Non-Clairvoyant} is $O(\sqrt{s} \cdot \log n)$-competitive.
This is only a $O(\log n)$ factor worse than {\sf Clairvoyant},
but does not require explicit knowledge of the conflict graph.
The algorithm is also greedy.
This algorithm uses as a subroutine a variation of the {\sf RandomizedRounds}
scheduling algorithm by Schneider and Wattenhofer \cite{Schneider09}
which uses randomized priorities and doesn't require knowledge of the conflict graph.

The $O(\sqrt{s})$ bound of Algorithm {\sf Clairvoyant} is actually tight.
Through a reduction from the graph coloring problem,
we show that it is impossible to approximate in polynomial time
any transactional scheduling
problem with $\beta = 1$ and $\ell = 1$
with a competitive ratio smaller than $O((\sqrt{s})^{1-\epsilon})$
for any constant $\epsilon > 0$,
unless {\sf NP$\subseteq$ZPP}.
To our knowledge, these results are significant improvements over the
best previously known bound of $O(s)$
for transactional memory contention managers.
For general workloads (including non-balanced workloads),
where transactions are equi-length ($\ell = O(1)$),
our analysis gives $O(s)$ competitive worst case bound,
since $\beta \geq 1/s$.
This bound matches the best previously known bound of $O(s)$
for general workloads.
The parametrization of $\beta$ that we provide
gives more tradeoffs and flexibility
for better scheduling performance,
as depicted by the performance of our algorithms
in balanced workloads.

\subsection{Related Work}\label{section:related}
Almost 10 year after publishing the seminal paper \cite{Herlihy93transactionalmemory}
to introduce the new research area of transactional memory,
Herlihy {\it et al.}~\cite{Her03} proposed Dynamic STM (DSTM)
for dynamic-sized data structures.
Later on, several other STM implementations have been proposed,
such as TL2 \cite{TL2}, TinySTM \cite{TinySTM}, and RSTM \cite{RSTM06} to name a few.
Among them, DSTM is the first practical obstruction-free\footnote{A synchronization mechanism
is obstruction-free if any thread that runs for a long time it eventually makes progress \cite{Obstruction-Free03}.}
implementation that seeks advice from the contention
manager module to either wait or abort a transaction at the time of conflict.

Several contention managers have been proposed in STM
and the performance of some of them has been
analyzed formally in \cite{Bimodal09,Attiya10,Gue05a,Gue05b,Schneider09,Window10}.
The first formal analysis of the performance of a contention manager
is given by Guerraoui {\it et al.}~\cite{Gue05a}
where they present the {\sf Greedy} contention manager
which decides in favor of older transactions using timestamps
and achieves $O(s^{2})$ competitive ratio.
This bound holds for any algorithm
which ensures the {\em pending commit} property (see Definition \ref{definition:pendingcommit}). 
Attiya {\it et al.}~\cite{Attiya10} improve the competitive ratio to $O(s)$,
and prove a matching lower bound of $\Omega(s)$
for any deterministic {\em work-conserving} algorithm
which schedules as many transactions as possible
(by choosing a maximal independent set of transactions).
The model in~\cite{Attiya10} is non-clairvoyant in the sense that
it requires no prior knowledge about the transactions
while they are executed.

Schneider and Wattenhofer \cite{Schneider09}
present a deterministic algorithm {\sf CommitBounds}
with competitive ratio $\Theta(s)$
and a randomized algorithm {\sf RandomizedRounds} with makespan
$O(C \log n)$ with high probability, for a set of $n$ transactions, 
where $C$ denotes the maximum number of  conflicts among transactions (assuming unit execution time durations for transactions).
Sharma {\it et al.}~\cite {Window10} study greedy contention managers
for $M\times N$ {\em execution windows of transactions}
with $M$ threads and $N$ transactions per thread and
present and analyze two new randomized greedy contention management algorithms.
Their first algorithm {\sf Offline-Greedy}
produces a schedule of length $O(\tau_{\max} \cdot (C+N\log(MN)))$
with high probability,
where $\tau_{\max}$ is the execution time duration of the longest transaction in the system,
and the second algorithm {\sf Online-Greedy}
produces a schedule of length $O(\tau_{\max} \cdot (C \log (MN) + N \log^2(MN)))$.
The competitiveness of both of the algorithms is within a poly-log factor
of $O(s)$.
Another recent work is {\sf Serializer} \cite{carstm}
which resolves a conflict by removing a conflicting transaction $T$
from the processor core where it was running,
and scheduling it on the processor core of the other transaction to which it conflicted with.
It is $O(n)$-competitive and in fact, it ensures that two transactions never conflict more than once.


{\em TM schedulers} \cite{Bimodal09,Preventing09,Yoo08,stealonabort09} offer an alternative approach to boost the TM performance.
A TM scheduler is a software component which decides when a particular transaction executes.
One proposal in this approach is {\sf Adaptive Transaction Scheduling (ATS)} \cite{Yoo08}
which measures adaptively the contention intensity of a thread,
and when the contention intensity increases beyond a threshold
it serializes the transactions.
The {\sf Restart} and {\sf Shrink} schedulers,
proposed by Dragojevi\'{c} {\em et al.}~\cite{Preventing09},
depend on the prediction of future conflicts
and dynamically serialize transactions based on the prediction to avoid conflicts.
The {\sf ATS}, {\sf Restart}, and {\sf Shrink} schedulers are $O(n)$-competitive.
{\sf Steal-On-Abort} \cite{stealonabort09} is yet another proposal
where the aborted transaction is given to the opponent transaction and queued behind it,
preventing the two transactions from conflicting again.

Recently, Attiya {\em et al.}~\cite{Bimodal09} proposed the {\sf BIMODAL} scheduler
which alternates between {\em writing epochs}
where it gives priority to writing transactions and {\em reading epochs}
where it gives priority to transactions that have issued only reads so far.
It achieves $O(s)$ competitive ratio on bimodal workloads
with equi-length transactions.
A bimodal workload contains only early-write and read-only transactions.

\paragraph{Outline of Paper.}
The rest of the paper is organized as follows.
We present our TM model and definitions in Section \ref{section:preliminaries}.
We present and formally analyze two new randomized algorithms, {\sf Clairvoyant} and {\sf Non-Clairvoyant},
in Sections \ref{section:offline} and \ref{section:online}, respectively.
The hardness result of balanced workload scheduling is presented in Section \ref{section:lower bound}.
Section \ref{section:conclusion} concludes the paper.

\section{Model and Definitions}
\label{section:preliminaries}

Consider a system of $n\geq 1$ threads $\cP = \{P_1, \cdots, P_n\}$
with a finite set  of $s$ shared resources $\cR =\{R_1,\ldots,R_s\}$.
We consider batch execution problems,
where the system issues a set of $n$ transactions $\cT=\{T_1, \cdots, T_n\}$ ({\em transaction workload}),
one transaction $T_i$ per thread $P_i$.
Each transaction is a sequence of actions (operations) each of which is either
a read or write to some shared resource.
The sequence of operations in a transaction must be {\em atomic}:
all operations of a transaction are guaranteed to either completely occur, or have no effects at all.
A transaction that only reads shared resources is called {\em read-only};
otherwise it is called a {\em writing} transaction.
We consider transaction workloads where at least one transaction is writing.

After a transaction is issued and starts execution it either {\em commits} or {\em aborts}.
A transaction that has been issued but not committed yet is said to be {\em pending}.
A pending transaction can {\em restart} multiple times until it eventually commits.
Concurrent write-write actions or read-write actions to shared objects
by two or more transactions cause conflicts between transactions.
If a transaction conflicts then it either aborts,
or it may commit and force to abort all other conflicting transactions.
In a {\em greedy schedule}, if a transaction aborts due to conflicts
it then immediately restarts and attempts to commit again.
We assume that the execution time advances synchronously
for all threads and a preemption and abort require negligible time. 
We also assume that
all transactions in the system are correct,
i.e., there are no faulty transactions.\footnote{A transaction is called faulty when it encounters
an illegal instruction producing a segmentation fault or experiences a page fault resulting to wait
for a long time for the page to be available \cite{Gue05b}.}

\begin{definition}[Pending Commit Property \cite{Gue05a}]
\label{definition:pendingcommit}
A contention manager obeys the {\em pending commit} property
if, whenever there are pending transactions,
some running transaction $T$ will execute uninterrupted until it commits.
\end{definition}

Let $\cR(T_i)$ denote the set of resources used by a transaction $T_i$.
We can write $\cR(T_i) = \cR_w(T_i) \cup \cR_r(T_i)$,
where $\cR_w(T_i)$ are the resources which are to be written by $T_i$,
and $\cR_r(T_i)$ are the resources to be read by $T_i$.

\begin{definition}[Transaction Conflict]
Two transactions $T_i$ and $T_j$ {\em conflict}
if at least one of them writes on a common resource,
that is, there is a resource $R$ such that
$R \in (\cR_w(T_i) \cap \cR(T_j)) \cup (\cR(T_i) \cap \cR_w(T_j))$
(we also say that $R$ causes the conflict).
\end{definition}

From the definition of transaction conflicts we can define the {\em conflict graph}
for a set of transactions.
In the conflict graph, each node corresponds to a transaction and each edge represents a conflict
between the adjacent transactions.

\begin{definition}[Conflict Graph]
\label{definition:3}
For a set of transactions $\cT$,
the {\em conflict graph} $G(\cT)=(V,E)$ has as nodes the transactions, $V = \cT$,
and $(T_i, T_j) \in E$
for any two transactions $T_i,T_j$ that conflict.
\end{definition}

Let $\gamma(R_j)$ denote the number of transactions that write resource $R_j$.
Let $\gamma_{\max} = \max_{j} \gamma(R_j)$.
Denote $\lambda_w(T_i) = |\cR_w(T_i)|$,
$\lambda_r(T_i) = |\cR_r(T_i)|$,
and $\lambda(T_i) = |\cR(T_i)|$,
the number of resources which are being accessed by transaction $T_i$ for write,
read, and both read and write.
Let $\lambda_{\max} = \max_i \lambda(T_i)$.
Note that in the conflict graph $G$
the maximum node degree is bounded by $\lambda_{\max} \cdot \gamma_{\max}$,
and also there is a node whose degree is at least $\gamma_{\max}$.

For any transaction $T_i$
we define the {\em balancing ratio}
$\beta(T_i) = \frac{|\cR_w(T_i)|}{|\cR(T_i)|}$
as the ratio of number of writes versus the total number of resources it accesses.
For a read-only transaction $\beta(T_i) = 0$.
For a writing transaction it holds $\frac{1}{s}\leq \beta(T_i) \leq  1$,
since there will be at least one write performed by $T_i$
to one of the $s$ resources.
We define the {\em global balancing ratio}
as the minimum of the individual writing transaction balancing ratios:
$\beta = \min_{(T_i \in \cT) \wedge (\lambda_w(T_i) > 0)} \beta(T_i)$.
We define {\em balanced transaction workloads} as follows
(recall that we consider workloads with at least one writing transaction):
\begin{definition}[Balanced Workloads]
We say that a workload (set of transactions) $\cT$ is {\em balanced}
if $\beta = \Theta(1)$.
\end{definition}
In other words, in balanced transaction workloads
the number of writes that each writing transaction performs
is a constant fraction of the total number of resource accesses (for read or write)
that the transaction performs.

Each transaction $T_{i}$ has execution time duration $\tau_{i} > 0$.
The execution time is the total number of discrete time steps that the
transaction requires to commit uninterrupted from the moment it starts.
In our model we assume that the execution time of each transaction is fixed.
Let $\tau_{max}= \max_i  \tau_i$ be the execution time of the longest transaction,
and $\tau_{min}=\min_i  \tau_i$ be the execution time of the shortest transaction.
We denote $\ell=\left\lceil\log\left(\frac{\tau_{max}}{\tau_{min}}\right)\right\rceil + 1$.
We finish this section
with the basic definitions of {\em makespan} and {\em competitive ratio}.

\begin{definition}[Makespan and Competitive Ratio]
Given a contention manager $\cA$ and a workload $\cT$,
$makespan_{\cA}(\cT)$ is the total time $\cA$ needs to commit all the transactions in $\cT$.
The {\em competitive ratio} is $CR_{\cA}(\cT)=\frac{makespan_{\cA}(\cT)}{makespan_{opt}(\cT)}$,
where {\sf opt} is the optimal off-line scheduler.
\end{definition}

\section{Clairvoyant Algorithm}\label{section:offline}

We describe and analyze Algorithm {\sf Clairvoyant}
(see Algorithm \ref{algorithm:clairvoyant}).
The writing transactions are divided into $\ell$ groups $A_0, A_1, \ldots, A_{\ell-1}$,
where $\ell = \left \lceil \log \left (\frac{\tau_{\max}}{\tau_{\min}} \right ) \right \rceil + 1$,
in such a way that $A_i$ contains transactions
with execution time duration in range
$[2^i \cdot \tau_{\min}, (2^{i+1} -1) \cdot \tau_{\min}]$,
for $0 \leq i \leq \ell - 1$.
Each group of transactions $A_i$
is then again divided into $\kappa$ subgroups
$A_{i}^0, A_{i}^1, \ldots, A_{i}^{\kappa-1}$,
where $\kappa = \lceil \log s \rceil + 1$,
such that each transaction $T \in A_{i}^j$
accesses (for read and write)
a number of resources in range $\lambda(T) \in [2^j, 2^{j+1}-1]$, for $0\leq j\leq \kappa-1$.
We assign an order to the subgroups
in such a way that  $A_i^j<A_k^l$ if $i<k$ or $i=k \wedge j<l$.
Note that some of the subgroups may be empty.
The read-only transactions are placed into a special group $B$
which has the highest order.

At any time $t$ the pending transactions are assigned a priority level
which determines which transactions commit or abort.
A transaction is assigned a priority which
is one of: {\em high} or {\em low}.
Let $\Pi_t^h$ and $\Pi_t^l$ denote the set of transactions
which will be assigned high and low priority, respectively, at time $t$.
In conflicts, high priority transactions abort low priority transactions.
Conflicts between transactions of the same priority level are resolved arbitrarily.
Suppose that $\wA_t$
is the lowest order subgroup that contains pending transactions at time $t$.
Only transactions from $\wA_t$ can be given high priority, that is $\Pi_t^h \subseteq \wA_t$.

The priorities are determined according to the conflict graph for the transactions.
Let $\cT_t$ denote the set of all transactions which are pending at time $t$.
(Initially, $\cT_0 = \cT$.)
Let $\cwT_t$ denote the pending transactions of $\wA_t$ at time $t$.
(Initially, $\cwT_0 = \wA_0$.)
Let $\wS_t$ denote the set of transactions in $\cwT_t$
which are pending and have started executing before $t$
but have not yet committed or aborted.
Let $\wS'_t$ denote the set of transactions in $\cT_t$
which conflict with $\wS_t$.
Let $\wI_t$ be a maximal independent set in the conflict graph $G(\cwT_t \setminus \wS'_t)$.
Then, the set of high priority transactions at time $t$ is
set to be $\Pi_t^h = \wI_t \cup \wS_t$.
The remaining transactions are given low priority,
that is, $\Pi_t^l = \cT_t \setminus \Pi_t^h$.
Note that the transactions in $\Pi_t^h$ do not conflict with each other.
The transactions $\Pi_t^h$
will remain in high priority in subsequent time steps $t' > t$ until they commit,
since the transactions in $\wS_{t'}$ are included in $\Pi_{t'}^h$.

This algorithm is clairvoyant in the sense that it requires explicit knowledge
of the various conflict relations at each time $t$.
The algorithm is greedy,
since at each time step each pending transaction is not idle.
The algorithm also satisfies the pending commit property
since at any time step $t$ at least one transaction from $\wA_t$
will execute uninterrupted until it commits.
We have assumed above that each transaction knows its execution length and
the number of resources it accesses.
Clearly, the algorithm computes the schedule in polynomial time.

\begin{algorithm}[t]
{\small
\KwIn{A set $\cT$ of $n$ transactions with global balancing ratio
$\beta$\;}
\KwOut{A greedy execution schedule\;}
\BlankLine
\nlset{-}  Divide writing transactions into $\ell=\lceil\log(\frac{\tau_{\max}}{\tau_{\min}})\rceil+1$
groups $A_0, A_1, \cdots, A_{\ell-1}$
in such a way that $A_i$ contains transactions with execution time duration
in range $[2^i \cdot \tau_{\min}, (2^{i+1}-1) \cdot \tau_{\min}]$;
Read-only transactions are placed in special group $B$\;

\nlset{-}  Divide $A_i$ again into $\kappa=\lceil\log s\rceil + 1$
subgroups $A_{i}^0,A_{i}^1, \cdots, A_{i}^{\kappa-1}$
in a way that each subgroup $A_{i}^j$
contains transactions that access a number of resource in the range
$[2^j,2^{j+1}-1]$\;

\nlset{-}  Order the groups and subgroups such that $A_i^j<A_k^l$ if $i<k$ or $i=k \wedge j<l$;
           special group $B$ has highest order\;

\BlankLine

\ForEach{time step $t = 0, ~1, ~2, ~3, \ldots$}{

\textbf{Set Definitions:}\\
{\Indp

$\cT_t$: set of transactions that are pending; ~~\tcp{$\cT_0 \gets \cT$}

$\wA_t$: lowest order group that contains pending transactions\;

$\cwT_t$: set of transactions in $\wA_t$ which are pending; ~~\tcp{$\cwT_0 \gets \wA_0$}

$\wS_t$: set of transactions in $\cwT_t$ which were started before $t$\;

$\wS'_t$: set of conflicting transactions in $\cT_t$ which conflict with $\wS_t$\;

$\wI_t:$ maximal independent set in the conflict graph $G(\cwT_t \setminus \wS'_t)$\;
}

\textbf{Priority Assignment:}\\
{\Indp
High priority transactions: $\Pi_t^h \gets \wI_t \cup \wS_t$\;
Low priority transactions: $\Pi_t^l \gets \cT_t \setminus \Pi_t^h$\;
}

\textbf{Conflict Resolution:}\\
{\Indp

Execute all pending transactions\;

\textbf{On conflict} of transaction $T_u$ with transaction $T_v$:\\
\Indp

\lIf {$(T_u \in \Pi_t^h) ~\wedge~ (T_v \in \Pi_t^l)$}
{$abort(T_u, T_v)$;} \lElse {$abort(T_v, T_u)$\;}
}
}
}
\caption{{\sf Clairvoyant}}
\label{algorithm:clairvoyant}
\end{algorithm}

\subsection{Analysis of Clairvoyant Algorithm}

We now give a competitive analysis of Algorithm {\sf Clairvoyant}.
Define $\tau_{\min}^j = 2^i \cdot \tau_{\min}$
and $\tau_{\max}^j = (2^{i+1} - 1) \cdot \tau_{\min}$.
Note that the duration of each transaction $T \in A_i^j$
is in range $[\tau_{\min}^j, \tau_{\max}^j]$,
and also $\tau_{\max}^j \leq 2 \tau_{\min}^j$.
Define $\lambda_{\min}^j = 2^j$ and $\lambda_{\max}^j = 2^{j+1}-1$.
Note that for each transaction $T \in A_i^j$,
$\lambda(T) \in [\lambda_{\min}^j, \lambda_{\max}^j]$,
and $\lambda_{\max}^j \leq 2 \lambda_{\min}^j$.
Let $\gamma_i^j(R_v)$ denote the number of transactions in a subgroup $A_i^j$
that write $R_v,  1\leq v\leq s$.
Let $\gamma_{\max}^j = max_{i \in [1,\ell], v \in [1,s]} \gamma_i^j(R_v)$.

In the next results we will first focus on a subgroup $A_i^j$
and we will assume that there are no other transactions in the system.
We give bounds for the competitive ratio for $A_i^j$
which will be useful when we later analyze the performance
for all the transactions in $\cT$.
\begin{lemma}
\label{observation:1}
If we only consider transactions in subgroup $A_i^j$,
then the competitive ratio is bounded by
$CR_{Clairvoyant}(A_i^j)\leq 2 \cdot \lambda_{\max}^j + 2.$
\end{lemma}

\begin{proof}
Since there is only one subgroup,
$\wA_t = A_i^j$.
A transaction $T \in A_i^j$
conflicts with at most $\lambda_{\max}^j \cdot \gamma_{\max}^j$
other transactions in the same subgroup.
If transaction $T$ is in low priority it is only
because some other conflicting transaction in $A_i^j$ is in high priority.
If no conflicting transaction is in high priority then $T$
becomes high priority immediately.
Since a high priority transaction executes uninterrupted until it commits,
it will take at most $\lambda_{\max}^j \cdot \gamma_{\max}^j$ time steps
until all conflicting transactions with $T$ have committed.
Thus, it is guaranteed that in at most $\lambda_{\max}^j \cdot \gamma_{\max}^j \cdot \tau_{\max}^j$
time steps $T$ becomes high priority.
Therefore,
$T$ commits by time $(\lambda_{\max}^j \cdot \gamma_{\max}^j + 1) \cdot \tau_{\max}^j$.
Since $T$ is an arbitrary transaction in $A_i^j$,
the makespan of the algorithm is bounded by:
$$makespan_{Clairvoyant}(A_i^j)\leq (\lambda_{\max}^j \cdot \gamma_{\max}^j+ 1) \cdot \tau_{\max}^j.$$
There is a resource that is accessed by at least $\gamma_{\max}^j$
transactions of $A_i^j$ for write.
All these transactions have to serialize
because they all conflict with each other in the common resource.
Therefore, the optimal makespan is bounded by:
$$makespan_{opt}(A_i^j) \geq \gamma_{\max}^j \cdot \tau_{\min}^j.$$

When we combine the upper and lower bounds we obtain a bound on the competitive ratio
of the algorithm:
\begin{eqnarray*}
CR_{Clairv.}(A_i^j)
 =  \frac{makespan_{Clairv.}(A_i^j)}{makespan_{opt}(A_i^j)}
 \leq  \frac{( \lambda_{\max}^j \cdot \gamma_{\max}^j  + 1 )\cdot \tau_{\max}^j}
{\gamma_{\max}^j \cdot \tau_{\min}^j}
 \leq  2 \cdot \lambda_{\max}^j + 2.
\end{eqnarray*}
\end{proof}

\begin{lemma}
\label{observation:2}
If we only consider transactions in subgroup $A_i^j$,
then the competitive ratio is bounded by
$CR_{Clairvoyant}(A_i^j) \leq 4 \cdot \frac{s / \beta }{\lambda_{\max}^j}.$
\end{lemma}

\begin{proof}
Since the algorithm satisfies the pending-commit property,
if a transaction $T \in A_i^j$ does not commit, then some
conflicting transaction $T'\in A_i^j$ must commit.
Therefore, the makespan of the algorithm is bounded by:
$$makespan_{Clairvoyant}(A_i^j) \leq |A_i^j| \cdot \tau_{\max}^j.$$

Each transaction in $T \in A_i^j$ accesses at least $\lambda_w(T)$ resources for write.
Since we only consider transactions in $A_i^j$,
$\lambda_{w}(T) \geq \beta \cdot \lambda_{\min}^j \geq \beta \cdot \lambda_{\max}^j /2$.
Consequently,
by the pigeonhole principle,
there will be a resource $R \in \cR$ which is accessed by at least
$\sum_{T \in A_i^j} \lambda_{w}(T) / s
\geq |A_i^j| \cdot \beta \cdot \lambda_{\max}^j / (2 s)$ transactions for write.
All these transactions accessing $R$ have to serialize because they conflict with each other.
Therefore,
the optimal makespan is bounded by:
$$makespan_{opt}(A_i^j)
\geq \frac{|A_i^j| \cdot \beta \cdot \lambda_{\max}^j} {2s}  \cdot \tau_{\min}^j.$$

When we combine the above bounds of the makespan we obtain the following bound
on the competitive ratio of the algorithm:
\begin{eqnarray*}
CR_{Clairvoyant}(A_i^j)
& = & \frac{makespan_{Clairvoyant}(A_i^j)} {makespan_{opt}(A_i^j)}
\leq  \frac{|A_i^j|\cdot \tau_{\max}^j} {\frac{|A_i^j| \cdot \beta \cdot \lambda_{\max}^j}{ 2 s} \cdot \tau_{\min}^j}
 \leq  4 \cdot \frac{s / \beta }{\lambda_{\max}^j}.
\end{eqnarray*}
\end{proof}

From Lemmas \ref{observation:1} and \ref{observation:2},
we obtain:

\begin{corollary}
\label{lemma:competitive-ratio-subgroup}
If we only consider transactions in subgroup $A_i^j$,
then the competitive ratio of the algorithm is bounded by
$CR_{Clairvoyant} (A_{i}^j) \leq 4 \cdot \min \left \{\lambda_{\max}^j, \frac{s/\beta}{\lambda_{\max}^j} \right \}$.
\end{corollary}

We now continue to provide a bound for the performance of individual groups.
This will help to provide bounds for all the transactions.

\begin{lemma}
\label{lemma:offline-competitive-group}
If we only consider transactions in group $A_i$,
then the competitive ratio of the algorithm is bounded by
$CR_{Clairvoyant} (A_{i})\leq 32 \cdot \sqrt{\frac{s}{\beta}}$.
\end{lemma}

\begin{proof}
Since $\lambda_{\max}^j=(2^{j+1}-1)$,
Corollary \ref{lemma:competitive-ratio-subgroup}
gives for each subgroup $A_{i}^j$ competitive ratio
$$CR_{Clairvoyant} (A_{i}^j)
\leq 4 \cdot \min \left \{2^{j+1}-1, \frac{s/\beta}{2^{j+1}-1} \right \}
\leq 8 \cdot \min \left \{2^{j}, \frac{s/\beta}{2^{j}} \right \}.$$
Let $\psi = \frac{\log(s/\beta)}{2}$.
Note that $\min \left \{2^{j}, \frac{s/\beta}{2^{j}} \right \} \leq 2^{j}$,
$\forall j \in [0, \lfloor \psi \rfloor]$;
and $\min \left \{2^{j}, \frac{s/\beta}{2^{j}} \right \} \leq \frac{s/\beta}{2^{j}} = 2^{2 \psi - j}$ ,
$\forall j \in [\lfloor \psi \rfloor + 1, \kappa-1]$.
Group $A_i$ contains $\kappa$ subgroups of transactions.
In the worst case, Algorithm {\sf Clairvoyant} will
commit the transactions in each subgroup according to their
order starting from the lowest order subgroup
and ending at the highest order subgroup,
since that's the order that the transactions are assigned a high priority.
Therefore,
\begin{eqnarray*}
CR_{Clairv.} (A_{i})
& \leq & \sum_{j=0}^{\kappa-1} CR_{Clairv.} (A_{i}^j)\\
& = &  \sum_{j=0}^{\lfloor \psi \rfloor} CR_{Clairv.} (A_{i}^j)
      + \sum_{j = \lfloor \psi \rfloor + 1}^{\kappa-1} CR_{Clairv.} (A_{i}^j) \\
& \leq & 8 \cdot \left ( \sum_{j=0}^{\lfloor \psi \rfloor} 2^{j}
        + \sum_{j = \lfloor \psi \rfloor + 1}^{k - 1} 2^{2 \psi - j} \right )
 \leq   8 \cdot \left ( 2 \cdot 2^{\psi}
        + 2 \cdot 2^{\psi} \right )
 =  32 \cdot \sqrt{\frac{s}{\beta}}.
\end{eqnarray*}
\end{proof}

\begin{theorem}[Competitive Ratio of {\sf Clairvoyant}]
\label{theorem:clairvoyant}
For set of transactions $\cT$,
Algorithm {\sf Clairvoyant} has competitive ratio
$CR_{Clairvoyant} (\cT) = O \left (\ell \cdot \sqrt{\frac{s}{\beta}} \right )$.
\end{theorem}

\begin{proof}
As there are $\ell$ groups of transactions $A_i$, and one group $B$,
in the worst case, Algorithm {\sf Clairvoyant} will
commit the transactions in each group according to their
order starting from the lowest order group
and ending at the highest order group.
Clearly, the algorithm will execute the read-only transactions in group $B$
in optimal time.
Therefore, using Lemma \ref{lemma:offline-competitive-group}, we obtain:
\begin{eqnarray*}
CR_{Clairvoyant}(\cT)
& \leq &  \sum_{i=0}^{\ell-1} CR_{Clairvoyant}(A_{i}) + CR_{Clairvoyant}(B)\\
& \leq & \sum_{i=0}^{\ell-1} 32 \cdot \sqrt{\frac{s}{\beta}} + 1
 =  32 \cdot  \ell \cdot  \sqrt{\frac{s}{\beta}} + 1.
\end{eqnarray*}
\end{proof}

The corollary below follows immediately from Theorem \ref{theorem:clairvoyant}.

\begin{corollary}[Balanced Workload]
\label{corollary:online-competitive-clairvoyant}
For balanced workload $\cT$ ($\beta=O(1)$) and when $\ell=O(1)$,
Algorithm {\sf Clairvoyant} has competitive ratio $CR_{Clairvoyant} (\cT) = O(\sqrt{s})$.
\end{corollary}

\section{Non-Clairvoyant Algorithm}\label{section:online}
We present and analyze Algorithm {\sf Non-Clairvoyant}
(see Algorithm \ref{algorithm:non-clairvoyant}).
This algorithm is similar to {\sf Clairvoyant} given at Section \ref{section:offline}
with the difference that the conflicts are resolved using priorities
which are determined without the explicit knowledge of the conflict graph.

Similar to Algorithm {\sf Clairvoyant},
the transactions are organized into groups and subgroups.
Lower order subgroups have always higher priority than
higher order subgroups.
At each time step $t$,
let $\wA_t$ denote the lowest order subgroup.
Clearly,
the transactions in $\wA_t$ have higher priority
than the transactions in all other subgroups,
and in case of conflicts only the transactions in $\wA_t$ win.
When transactions in the same subgroup conflict,
the conflicts are resolved according to random priority numbers.
When a transaction starts execution it chooses uniformly at random
a discrete number $r(T) \in [1, n]$.
In case of a conflict of transaction $T_w$
with another transaction $T_x$ in the same subgroup
with $r(T_x) < r(T_w)$,
then $T_x$ aborts $T_w$,
and otherwise $T_w$ aborts $T_x$.
When transaction $T_w$ restarts, it cannot abort $T_x$
until $T_x$ has been committed or aborted.
After every abort, the newly started transaction chooses again
a new discrete number uniformly at random in the interval $[1, n]$.
The idea of randomized priorities has been introduced originally by
Schneider and Wattenhofer \cite{Schneider09}
in their Algorithm {\sf RandomizedRounds}.

This algorithm is non-clairvoyant in the sense that it does not depend
on knowing explicitly the conflict graph to resolve conflicts.
The algorithm is greedy but does have the pending commit property.
The groups and subgroups can be implemented in the algorithm
since we assume that each transaction knows
its execution time and the number of resources that it accesses.
Clearly, the algorithm computes the schedule in polynomial time.

\begin{algorithm}[t]
{\small
\KwIn{A set $\cT$ of $n$ transactions with global balancing ratio
$\beta$\;}
\KwOut{A greedy execution schedule\;}
\BlankLine
\nlset{-} Divide transactions into $\ell=\lceil\log(\frac{\tau_{\max}}{\tau_{\min}})\rceil+1$
groups $A_0, A_1, \cdots, A_{\ell-1}$
in such a way that $A_i$ contains transactions with execution time duration
in range $[2^i \cdot \tau_{\min}, (2^{i+1}-1) \cdot \tau_{\min}]$;
Read-only transactions are placed in special group $B$\;

\nlset{-} Divide $A_i$ again into $\kappa=\lceil\log s\rceil + 1$
subgroups $A_{i}^0,A_{i}^1, \cdots, A_{i}^{\kappa-1}$
in a way that each subgroup $A_{i}^j$
contains transactions that access a number of resource in the range
$[2^j,2^{j+1}-1]$\;

\nlset{-} Order the groups and subgroups such that $A_i^j<A_k^l$ if $i<k$ or $i=k \wedge j<l$;
special group $B$ has highest order\;

\BlankLine

\ForEach{time step $t = 0, ~1, ~2, ~3, \ldots$}{

Execute all pending transactions; \tcp{at $t=0$ issue all transactions}

\textbf{On (re)start} of transaction $T$:\\
{\Indp $r(T) \gets $ random integer in $[1,n]$\;}

\textbf{On conflict} of transaction $T_u \in A_i^j$ with transaction $T_v \in A_k^l$:\\
{\Indp

\lIf{$A_i^j < A_k^l$}{$abort(T_u, T_v)$\;}
\lElseIf{$A_i^j > A_k^l$}{$abort(T_v, T_u)$\;}

\Indp

\lElseIf{$r(T_u) < r(T_v)$} {$abort(T_u, T_v)$}  \tcp*{The case $A_i^j = A_k^l$}

\Indp

\lElse{$abort(T_v, T_u)$\;}

\tcp{In case a transaction $T_u$ aborts $T_v$ because $r(T_u) < r(T_v)$,
then when $T_v$ restarts it cannot abort $T_u$ until $T_u$ commits or aborts}
}
}
}
\caption{{\sf Non-Clairvoyant}}
\label{algorithm:non-clairvoyant}
\end{algorithm}

\subsection{Analysis of Non-Clairvoyant Algorithm}

In the analysis given below,
we study the properties of Algorithm {\sf Non-Clairvoyant} and give its competitive ratios.
We use the following adaptation
of the response time analysis of Algorithm {\sf RandomizedRounds} given in \cite{Schneider09}.
It uses the following Chernoff bound:

\begin{lemma}[Chernoff Bound]\label{lemma:chernoff2}
Let $X_1, X_2, \ldots, X_n$ be independent Poisson trials
such that, for $1 \leq i \leq n$, ${\bf Pr}(X_i = 1) = pr_i$,
where $0 < pr_i < 1$.
Then, for $X = \sum_{i=1}^{n} X_i$,
$\mu = {\bf E}[X] = \sum_{i=1}^{n} pr_i$,
and any $0 < \delta \leq 1, {\bf Pr}(X < (1-\delta) \mu) < e^{-\delta^2 \mu / 2}.$
\end{lemma}

\begin{lemma}[Adaptation from Schneider and Wattenhofer \cite{Schneider09}]
\label{lemma:roger}
Given a transaction scheduling problem with $n$ concurrent transactions,
where each transaction has execution time at most $\tau$,
the time span a transaction $T$ needs from the moment it is issued until
commit is $16 \cdot e \cdot (d_{T}+1) \cdot \tau \cdot \ln n$
with probability at least $1-\frac{1}{n^2}$,
where $d_T$ is the number of transactions conflicting with $T$.
\end{lemma}

\begin{proof}
Consider the respective conflict graph $G$ of the problem with the $n$ transaction.
Let $N_T$ denote the set of conflicting transactions for $T$
(these are the neighbors of $T$ in $G$).
Let $r(T)$ denote the random priority number choice of $T$ in range $[1,n]$.
The probability that for transaction $T$
no transaction $T' \in N_{T}$ has the same random number is:
$${\bf Pr}(\nexists T' \in N_{T} | r(T) = r(T'))= \left (1-\frac{1}{n} \right)^{d_{T}} \geq \left ( 1-\frac{1}{n} \right )^n \geq \frac{1}{e}.$$
The probability that $r(T)$ is at least as small as $r(T')$
for any transaction $T' \in N_{T}$ is $\frac{1}{d_{T}+1}$.
Thus, the chance that $r(T)$ is smallest and different among all its neighbors in $N_T$
is at least $\frac{1}{e \cdot(d_{T}+1)}$.
If we conduct $16 \cdot e \cdot (d_{T}+1) \cdot \ln n$ trials,
each having success probability $\frac{1}{e \cdot (d_{T}+1)}$,
then the probability that the number of successes $Z$ is less than $8 \ln n$ becomes:
${\bf Pr}(Z < 8 \cdot \ln n) <e^{-2\cdot \ln n} = 1 / n^2$,
using the Chernoff bound of Lemma \ref{lemma:chernoff2}.
Since every transaction has execution time at most $\tau$,
the total time spent until a transaction commits is
at most $16 \cdot e \cdot (d_{T}+1) \cdot \tau \cdot \ln n$,
with probability at least $1 - 1/n^2$.
\end{proof}

We now give competitive bounds for some subgroup $A_i^j$
and later extend the results to all the transactions in $\cT$.
The proofs are similar as in the analysis of Algorithm {\sf Clairvoyant}
and can be found in the appendix.

\begin{lemma}
\label{observation:3}
If we only consider transactions in subgroup $A_i^j$,
then the competitive ratio is bounded by
$CR_{Non-Clairvoyant}(A_i^j)\leq 64 \cdot e \cdot \lambda_{\max}^j \cdot \ln n$
with probability at least $1-\frac{|A_i^j|}{n^2}.$
\end{lemma}


\begin{lemma}
\label{observation:4}
If we only consider transactions in subgroup $A_i^j$,
then the competitive ratio is bounded by
$CR_{Non-Clairvoyant}(A_i^j)\leq 64 \cdot e \cdot \frac{s/\beta}{\lambda_{\max}^j} \cdot \ln n $
with probability at least $1-\frac{|A_i^j|}{n^2}.$
\end{lemma}


From Lemmas \ref{observation:3} and \ref{observation:4}, we obtain:

\begin{corollary} \label{lemma:competitive-ratio-subgroup-non}
If we only consider transactions in subgroup $A_i^j$,
then the competitive ratio of the algorithm is bounded by
$CR_{Non-Clairvoyant} (A_{i}^j) \leq 64 \cdot e\cdot\min \left\{\lambda_{\max}^j, \frac{s/\beta}{\lambda_{\max}^j}\right\}\cdot \ln n$
with probability at least $1-\frac{|A_i^j|}{n^2}.$
\end{corollary}

We now provide a bound for the performance of individual groups
which will help to provide bounds for all the transactions.

\begin{lemma}
\label{lemma:online-competitive-group}
If we only consider transactions in group $A_i$,
then the competitive ratio of the algorithm is bounded by
$CR_{Non-Clairvoyant} (A_{i}) \leq 512 \cdot e \cdot \sqrt{\frac{s}{\beta}} \cdot \ln n$
with probability at least $1-\frac{|A_i|}{n^2}.$
\end{lemma}


\begin{theorem}[Competitive Ratio of {\sf Non-Clairvoyant}]
\label{theorem:online-competitive}
For a set of transactions $\cT$, Algorithm {\sf Non-Clairvoyant}
has competitive ratio
$CR_{Non-Clairvoyant}(\cT)=O\left (\ell  \cdot \sqrt{\frac{s}{\beta}} \cdot \log n\right )$
with probability at least $1 - \frac{1} {n}$.
\end{theorem}


The corollary below follows immediately from Theorem \ref{theorem:online-competitive}.

\begin{corollary}[Balanced Workload]
\label{corollary:online-competitive-equi}
For balanced workload $\cT$ ($\beta=O(1)$) and when $\ell=O(1)$,
Algorithm {\sf Non-Clairvoyant} has competitive ratio
$CR_{Non-Clairvoyant}(\cT)=O(\sqrt{s} \cdot \log n)$
with probability at least $1- \frac {1} {n}$.
\end{corollary}

\section{Hardness of Balanced Transaction Scheduling}
\label{section:lower bound}
In this section, we show that the performance of
{\sf Clairvoyant} is tight
by reducing the graph coloring problem
to the transaction scheduling problem.

A \textsc{Vertex Coloring} problem instance
asks whether a given graph $G$ is $k$-colorable \cite{Vertex90}.
A valid $k$-coloring
is an assignment of integers $\{1,2,\cdots, k\}$ (the colors)
to the vertices of $G$ so that neighbors receive different integers.
The chromatic number, $\chi(G)$ is the smallest $k$ such that $G$
has a valid $k$-coloring.
We say that an algorithm approximates $\chi(G)$ with approximation ratio
$q(G)$ if it outputs $u(G)$ such that $\chi(G)\leq u(G)$
and $u(G)/\chi(G) \leq q(G)$.
Typically, $q(G)$ is expressed only as a function of $n$,
the number of vertices in $G$.
It is well known that known \textsc{Vertex Coloring} is NP-complete.
It is also shown in \cite {Chromatic96}
that unless \textsf{NP$\subseteq$ZPP},
there does not exist
a polynomial time algorithm to approximate $\chi(G)$ with approximation ratio
$O(n^{1-\epsilon})$ for any constant $\epsilon > 0$,
where $n$ denotes the number of vertices in graph $G$.

A \textsc{Transaction Scheduling} problem instance
asks whether a set of transactions $\cT$ with a set of resources $\cR$
has makespan $k$ time steps.
We give a polynomial time reduction of the
\textsc{Vertex Coloring} problem to the \textsc{Transaction Scheduling} problem.
Consider an input graph $G = (V,E)$ of the \textsc{Vertex Coloring} problem,
where $|V| = n$ and $|E| = s$.
We construct a set of transactions $\cT$ such that for each $v \in V$
there is a respective transaction $T_v \in \cT$;
clearly, $|\cT| = |V| = n$.
We also use a set of resources $\cR$ such that for each edge $e \in V$
there is a respective resource $R_e \in \cR$;
clearly, $|\cR| = |E| = s$.
If $e = (u,v) \in E$,
then both the respective transactions $T_u$ and $T_v$
use the resource $R_e$ for write.
Since all transaction operations are writes,
we have that $\beta = 1$.
We take all the transactions to have the same execution
length equal to one time step,
that is, $\tau_{\max} = \tau_{\min} = 1$, and $\ell = 1$.

Let $G'$ be the conflict graph for the transactions $\cT$.
Note that $G'$ is isomorphic to $G$.
Node colors in $G$ correspond to time steps
in which transactions in $G'$ are issued.
Suppose that $G$ has a valid $k$-coloring.
If a node $v \in G$ has a color $x$,
then the respective transaction $T_v \in G'$
can be issued and commit at time step $x$,
since no conflicting transaction (neighbor in $G'$)
has the same time assignment (color) as $T_v$.
Thus, a valid $k$-coloring in $G$ implies a schedule with makespan $k$
for the transactions in $\cT$.
Symmetrically, a schedule with makespan $k$ for $\cT$
implies a valid $k$-coloring in $G$.

It is easy to see that the problem \textsc{Transaction Scheduling} is in $NP$.
From the reduction of the \textsc{Vertex Coloring} problem,
we also obtain that \textsc{Transaction Scheduling} is $NP$-complete.

From the above reduction,
we have that an approximation ratio $q(G)$ of the \textsc{Vertex Coloring} problem
implies the existence of a scheduling algorithm $\cA$
with competitive ratio $CR_{\cA}(\cT) = q(G)$ of the respective
\textsc{Transaction Scheduling} problem instance,
and vice-versa.
Since $s = |\cR| = |E| \leq n^2$,
an $(\sqrt{s})^{1-\epsilon}$ competitive ratio of $\cA$
implies at most an $n^{1-\epsilon}$ approximation ratio of \textsc{Vertex Coloring}.
Since, we know that unless \textsf{NP$\subseteq$ZPP},
there does not exist
a polynomial time algorithm to approximate $\chi(G)$ with approximation ratio
$O(n^{1-\epsilon})$ for any constant $\epsilon > 0$,
we obtain a symmetric result for the \textsc{Transaction Scheduling} problem:

\begin{theorem}[Approximation Hardness of \textsc{Transaction Scheduling}]
\label{theorem:approximation-hardness}
Unless {\sf NP$\subseteq$ZPP},
we cannot obtain a polynomial time transaction scheduling algorithm
such that for every input instance with $\beta = 1$ and $\ell = 1$
of the \textsc{Transaction Scheduling}
problem the algorithm achieves competitive ratio smaller than $O((\sqrt{s})^{1-\epsilon})$
for any constant $\epsilon > 0$.
\end{theorem}

Theorem \ref{theorem:approximation-hardness} implies that the $O(\sqrt s)$
bound of Algorithm {\sf Clairvoyant},
given in Corollary \ref{corollary:online-competitive-clairvoyant}
for $\beta = O(1)$ and $\ell = O(1)$, is tight.

\section{Conclusions}
\label{section:conclusion}
We have studied the competitive ratios achieved by transactional contention managers
on balanced workloads. The randomized algorithms presented in this paper allow to
achieve best competitive bound on balanced workloads.
We also establish hardness results on the competitive ratios
in our balanced workload model by reducing the well known NP-complete
vertex coloring problem
to the transactional scheduling problem.

There are several interesting directions for future work.
As advocated in~\cite{Her03},
our algorithms are conservative $-$ abort at least one transaction involved in a conflict $-$
as it reduces the cost to track conflicts and dependencies.
It is interesting to look whether the other schedulers
which are less conservative can give improved competitive ratios by reducing the overall makespan.
First, our study can be complemented by studying other performance measures,
such as the average response time of transactions under balanced workloads.
Second, while we have theoretically analyzed the behavior of balanced workloads,
it is interesting to see how our contention managers compare experimentally with prior
transactional contention managers, e.g., \cite{carstm,Yoo08,Gue05a,stealonabort09}.

\bibliographystyle{splncs03}
\bibliography{balanced-OPODIS2010}

\newpage
\appendix

\section{Proofs of Section \ref{section:online}}

\noindent
{\bf Proof of Lemma \ref{observation:3}:}
\begin{proof}
Since there is only one subgroup,
a transaction $T \in A_i^j$
conflicts with at most $d_T \leq \lambda_{\max}^j \cdot \gamma_{\max}^j$
other transactions in the same subgroup.
From Lemma \ref{lemma:roger},
it will take at most
$x = 16 \cdot e \cdot (\lambda_{\max}^j \cdot \gamma_{\max}^j + 1) \cdot \tau_{\max}^j \cdot \ln n$
time steps until $T$ commits,
with probability at least $1-\frac{1}{n^2}.$
Considering now all the transactions in $A_i^j$,
and taking the union bound of individual event probabilities,
we have that all the transactions in $A_i^j$ commit
within time $x$ with probability
at least $1-\frac{|A_i^j|}{n^2}.$
Therefore,
with probability at least $1-\frac{|A_i^j|}{n^2}$,
the makespan is bounded by:
$$makespan_{Non-Clairvoyant}(A_i^j)
\leq 16 \cdot e \cdot (\lambda_{\max}^j \cdot \gamma_{\max}^j + 1) \cdot \tau_{\max}^j \cdot \ln n.$$

Similar to Lemma \ref{observation:1}, there is a resource that is accessed by at least $\gamma_{\max}^j$
transactions of $A_i^j$ for write
so that all these transactions have to be serialized because of the conflicts.
Therefore, the optimal makespan is bounded by:
$$makespan_{opt}(A_i^j) \geq \gamma_{\max}^j \cdot \tau_{\min}^j.$$

By combining the upper and lower bounds,
we obtain a bound on the competitive ratio:
\begin{eqnarray*}
CR_{Non-Clairvoyant}(A_i^j)
& = & \frac{makespan_{Non-Clairvoyant}(A_i^j)}{makespan_{opt}(A_i^j)}\\
& \leq & \frac{16 \cdot e \cdot (\lambda_{\max}^j \cdot \gamma_{\max}^j  + 1 ) \cdot \tau_{\max}^j \cdot \ln n}
{\gamma_{\max}^j \cdot \tau_{\min}^j}\\
& \leq & 32 \cdot e \cdot (\lambda_{\max}^j + 1) \cdot \ln n \\
& \leq & 64 \cdot e \cdot \lambda_{\max}^j \cdot \ln n,
\end{eqnarray*}
with probability at least $1-\frac{|A_i^j|}{n^2}.$
\end{proof}

\noindent
{\bf Proof of Lemma \ref{observation:4}:}
\begin{proof}
Since for any transaction $T \in A_i^j$,
$d_T \leq |N_T| \leq |A_i^j| - 1$,
similar to the proof of Lemma \ref{observation:3},
with probability at least $1-\frac{|A_i^j|}{n^2},$
the makespan is bounded by:
$$makespan_{Non-Clairvoyant}(A_i^j) \leq 16 \cdot e \cdot |A_i^j| \cdot \tau_{\max}^j \cdot \ln n.$$

Similar to Lemma \ref{observation:2},
the optimal makespan is bounded by:
$$makespan_{opt}(A_i^j)
\geq \frac{|A_i^j| \cdot \beta \cdot \lambda_{\max}^j} {2s}  \cdot \tau_{\min}^j.$$

When we combine the above bounds of the makespan we obtain a bound
on the competitive ratio:
\begin{eqnarray*}
CR_{Non-Clairvoyant}(A_i^j)
& = & \frac{makespan_{Non-Clairvoyant}(A_i^j)} {makespan_{opt}(A_i^j)}\\
& \leq & \frac{16\cdot e\cdot |A_i^j| \cdot \tau_{\max}^j \cdot \ln n} {\frac{|A_i^j| \cdot \beta \cdot \lambda_{\max}^j}{ 2 s} \cdot \tau_{\min}^j}
 \leq  64 \cdot e \cdot \frac{s / \beta }{\lambda_{\max}^j} \cdot \ln n,
\end{eqnarray*}
with probability at least $1-\frac{|A_i^j|}{n^2}.$
\end{proof}

\noindent
{\bf Proof of Lemma \ref{lemma:online-competitive-group}:}
\begin{proof}
Since $\lambda_{\max}^j=(2^{j+1}-1)$,
Corollary \ref{lemma:competitive-ratio-subgroup-non}
gives for each subgroup $A_{i}^j$ competitive ratio
\begin{eqnarray*}
CR_{Non-Clairvoyant} (A_{i}^j)
& \leq & 64 \cdot e \cdot \min \left \{2^{j+1}-1, \frac{s/\beta}{2^{j+1}-1}  \right \}\cdot \ln n\\
& \leq & 128 \cdot e \cdot \min \left \{2^{j}, \frac{s/\beta}{2^{j}}  \right \}\cdot \ln n,
\end{eqnarray*}
with probability at least
$1-\frac{|A_i^j|}{n^2}.$
Following the proof steps as in Lemma \ref{lemma:offline-competitive-group},
we obtain:
\begin{eqnarray*}
CR_{Non-Clairvoyant} (A_{i})
& \leq & 512 \cdot e \cdot \sqrt{\frac{s}{\beta}} \cdot \ln n.
\end{eqnarray*}
This bound holds with
with probability at least
$1-\frac{\sum_{j=0}^{\kappa-1}|A_i^j|}{n^2}=1-\frac{|A_i|}{n^2}$,
since $\sum_{j=0}^{\kappa-1} |A_i^j| = |A_i|$.
\end{proof}

\noindent
{\bf Proof of Theorem \ref{theorem:online-competitive}:}
\begin{proof}
As there are $\ell$ groups of transactions $A_i$, and one group $B$,
in the worst case, Algorithm {\sf Non-Clairvoyant} will
commit the transactions in each group according to their
order starting from the lowest order group
and ending at the highest order group.
Clearly, the algorithm will execute the read-only transactions in group $B$
in optimal time.
Therefore, using Lemma \ref{lemma:online-competitive-group}
we obtain:
\begin{eqnarray*}
CR_{Non-Clairvoyant}(\cT)
& \leq & \sum_{i=1}^{\ell} CR_{Non-Clairvoyant}(A_{i}) + CR_{Non-Clairvoyant}(B)\\
& \leq & \sum_{i=0}^{\ell-1} 512 \cdot e \cdot \sqrt{\frac{s}{\beta}}\cdot \ln n + 1\\
& = & 512 \cdot e \cdot \ell \cdot \sqrt{\frac{s}{\beta}} \cdot  \ln n + 1,
\end{eqnarray*}
with probability at least  $1-\frac{\sum_{i=0}^{\ell-1} |A_i|}{n^2}=1-n^{-1}$,
since $\sum_{i=0}^{\ell-1} |A_i| = |\cT| = n$.
\end{proof}

\end{document}